\documentclass[11pt, oneside]{article}   	% use "amsart" instead of "article" for AMSLaTeX format
\usepackage{geometry}                		% See geometry.pdf to learn the layout options. There are lots.
\usepackage{authblk}
\usepackage{graphicx}				% Use pdf, png, jpg, or eps� with pdflatex; use eps in DVI mode
\usepackage{amssymb, amsthm, mathtools,bbm}
\usepackage[square,sort,comma,numbers]{natbib}
\usepackage{color}

\newtheorem{theorem}{Theorem}
\newtheorem{proposition}{Proposition}

\newcommand{\tr}{{\operatorname{Tr}\,}}
\newcommand{\fs}{-}

\title{Trajectory phase transitions in non-interacting systems: all-to-all dynamics and the random energy model}

\author[1,2]{Juan P. Garrahan} 
\author[3,4]{Chokri Manai}
\author[3,4,5]{Simone Warzel}
\affil[1]{\small School of Physics and Astronomy, University of Nottingham, Nottingham, NG7~2RD, UK}
\affil[2]{\small Centre for the Mathematics and Theoretical Physics of Quantum Non-Equilibrium Systems,
University of Nottingham, Nottingham, NG7 2RD, UK}
\affil[3]{\small Department of Mathematics, TU Munich, Germany}
\affil[4]{\small Munich Center for Quantum Science and Technology, Munich, Germany}
\affil[5]{\small Department of Physics, TU Munich, Germany}
\date{\today}							% Activate to display a given date or no date

\begin{document}
\maketitle
\begin{abstract}
	We study the fluctuations of time-additive random observables in the stochastic dynamics of a system of $N$ non-interacting Ising spins. We mainly consider the case of all-to-all dynamics where transitions are possible between any two spin configurations with uniform rates. We show that the cumulant generating function of the time-integral of a normally distributed quenched random function of configurations, i.e., the energy function of the random energy model (REM),
	has a phase transition in the large $N$ limit for trajectories of any time extent. 
	We prove this by determining the exact limit of the scaled cumulant generating function. This is accomplished by connecting the dynamical problem to a spectral analysis of the all-to-all quantum REM. We also discuss finite $N$ corrections as observed in numerical simulations. 

\end{abstract}

\bigskip

\section{Introduction}

In statistical mechanics we are used to studying static phase transitions from singularities in partition sums \cite{Chandler1987}: the value of a control parameter at which the free-energy becomes non-analytic (in the infinite-size limit) indicates that the equilibrium ensemble of configurations undergoes a phase change. The standard equilibrium ensemble method can be generalised straightforwardly to stochastic dynamics by replacing configurations with trajectories, static observables with (time-extensive) functions of trajectories, and the partition sum with the corresponding moment generating function of the trajectory observable \cite{Ruelle2004,Lecomte2007}. The ``thermodynamics of trajectories'' approach \cite{Merolle2005, Garrahan2018, Jack2020} allows to study dynamical or ``trajectory'' phase transitions, that is, singular changes in the nature of dynamical fluctuations that often are not reflected in (thermo)static properties or occur at different parameters of the model. The singularities of the relevant large deviation (LD) functions \cite{Touchette2009} reveal phase transitions in, for example,  the dynamical activity of glassy  systems \cite{Garrahan2007, Hedges2009, Speck2012}, in time-integrated currents in exclusion processes \cite{Derrida2007,Appert-Rolland2008,Jack2015}, and in (active) work in active matter \cite{Nemoto2019}. 

An interesting question is what occurs in a system of many degrees of freedom whose dynamics is non-interacting when one considers the fluctuations of a (quenched) random trajectory-observable that couples them. Our main object of interest will be a system of $N$ Ising spins which all flip independently from each other. For the case of non-random {\em local} observables and independent spins with single spin-flip dynamics recent results \cite{Vasiloiu2020} show that in certain cases there is a phase transition in the LD function. While, naively, one might expect nothing interesting to occur due to the non-interacting nature of the dynamics, these results indicate that the optimal way to generate  large fluctuations is by means of effectively highly correlated dynamics which is singularly different from the typical dynamics \cite{Nyawo2016}. 

Here we start addressing the problem of {\em random and long-ranged} trajectory observables by considering the time integral of a function of configurations whose values are normally distributed with zero mean and variance $N$, that is, the energy function of the simplest mean-field spin glass: the random energy model (REM) \cite{Derrida:1980mg,Bov06}. For simplicity we will consider dynamics which is all-to-all, that is, allowed configuration changes are those where any number of spins can flip simultaneously and independently. We also comment on the case of single-spin flips, which corresponds to the quantum random energy model (QREM).

The general problem we consider here has relevance in several areas. One is the minimisation via trajectory sampling of (quasi) random cost functions \cite{Mair2022}, which arises for example when training neural networks. A second one is in connection to measurement induced phase transitions in quantum systems \cite{Li2018,Skinner2019}, where the calculation of Renyi entropies reduces to computing  the optimal 
dynamics of a random coupling function~ \cite{Agrawal2021,Altland2021} in a system of classical replicas which evolve independently.

\section{Unbiased dynamics} 

Any continuous-time Markov process with trajectories $ \pmb{\omega}: [0,\infty) \to \mathcal{Q}_N  $ on the configuration space $ \mathcal{Q}_N  := \{ -1,1 \}^N  $of $ N $ Ising spins is uniquely characterised in terms of the transition rates $ w_{\pmb{\sigma}\to \pmb{\tau}} $ of spin configurations $ \pmb{\sigma} = (\sigma_1, \dots , \sigma_N) \in \mathcal{Q}_N $ to any other configuration $ \pmb{\tau} $, and the associated escape rates $ r_{\pmb{\sigma}} := \sum_{ \pmb{\tau} \neq \pmb{\sigma}}  w_{\pmb{\sigma}\to \pmb{\tau}}  $. The latter governs the law, $  r_{\pmb{\sigma}}  e^{-  r_{\pmb{\sigma}}  \Delta t} $, of the sojourn time $ \Delta t $ until the next jump out of $ \pmb{\sigma}  $.  
In the following, we choose $  w_{\pmb{\sigma}\to \pmb{\tau}}  := N 2^{-N}  $ independent of the configuration. Since the connectivity of this jump process is then described by the complete graph on $ 2^N $ vertices (i.e.\ spin configurations), this dynamics is called the completely connected or all-to-all stochastic dynamics on Ising configurations. 
Using Dirac's notation, in which $ \left\{ | \pmb{\sigma}\rangle  \ | \   \pmb{\sigma} \in  \mathcal{Q}_N  \right\} $ stands for the canonical orthonormal basis in the Hilbert space $ \ell^2(\mathcal{Q}_N) \equiv \otimes_{j=1}^N \mathbb{C}^2 $,  the generator of this Markov process is given by
$$
W := 
	\sum_{\substack{ 
		\pmb{\sigma}, \pmb{\tau} \in  \mathcal{Q}_N \\ \pmb{\sigma} \neq \pmb{\tau}  
		}}  
		w_{\pmb{\sigma}\to \pmb{\tau}} \, | \pmb{\tau}\rangle\langle  \pmb{\sigma} | 
		-  \sum_{\pmb{\sigma}\in  \mathcal{Q}_N} r_{\pmb{\sigma}}\,   | \pmb{\sigma}\rangle\langle  \pmb{\sigma} |  
		= 
		N \left(  | \fs \rangle \langle \fs |  - \mathbbm{1} \right) 
$$
in terms of the orthogonal projection $  | \fs \rangle \langle \fs |  $ onto the 'flat state'  defined by $\langle \pmb{\sigma}  | \fs \rangle  = 2^{-N/2} $. In its probabilistic interpretation, $ W $ is considered an operator on $ \ell^1( \mathcal{Q}_N) $ and acts on probability distributions $ | p_t \rangle $, i.e.\ $p_t(\pmb{\sigma} ) \equiv  \langle \pmb{\sigma}  | p_t \rangle \geq 0 $ and $  \sum_{\pmb{\sigma} \in \mathcal{Q}_N} p_t(\pmb{\sigma} )  = 1 $. The dynamics of any initial distribution is governed by the master equation
$$
\partial_t | p_t \rangle = W | p_t \rangle . 
$$
The completely connected stochastic dynamics can be regarded as a further simplification of the dynamics of independent spin flips at infinite temperature. The latter is generated by
$\widehat W := \sum_{j=1}^N \left( X_j -\mathbbm{1} \right)$, 
in terms of the Pauli-$ X $ matrices, which flip the $ j $th spin, i.e.\ $  X_j |  \pmb{\sigma}  \rangle = | \sigma_1, \dots , - \sigma_j , \dots , \sigma_N \rangle $.  
Both Markov processes are irreducible and share the equidistribution $ p_{\rm ss}(\pmb{\sigma} ) := 2^{-N} $ as its unique  invariant measure. One difference is their spectral gap, which governs the rate of approach to the equidistribution. While the spectral gap is $ N $ in the case of $ W $, it is $ 2 $ in the case of $ \widehat W  $. In this paper we focus on the completely connected dynamics  $ W $ and only comment on the single spin-flip dynamics $ \widehat W  $.

The dynamics generated by $W$ (and $ \widehat W  $) is ``infinite temperature'' in the sense that transitions are completely independent of the initial and final states. The operator $W$ is therefore bi-stochastic, 
$\langle - | W = 0,  \, W | - \rangle = 0$, with the first equality indicating conservation of probability, and the second that the stationary state is also the flat state (the stationary probability vector being $ 2^{-N/2} |-\rangle$). Since the dynamics of all spins is independent, all correlation functions are unconnected.

\section{Trajectory observable and REM}

We study the statistics under the above defined 
all-to-all independent dynamics of a trajectory observable chosen to explore the energy landscape of the REM \cite{Derrida:1980mg,Bov06}.
The REM,
$U: \mathcal{Q}_N \to \mathbb{R}$, is a Gaussian random field (with randomness independent of the Markov process) in which the values $ U(\pmb{\sigma}) $ are distributed independently for all $ \pmb{\sigma} \in \mathcal{Q}_N $ with identical normal law uniquely characterised by zero mean and covariance $ N $. The units are chosen so that the REM's large deviations occur on order $ N $ which agrees with the norm of $ W $.  In this context, we recall~\cite{LLR83,Bov06} that the REM's minimum (and similarly for its maximum) satisfies the extremal value statistics:
\begin{equation}\label{eq:maxU}
\mathbb{P}\left( \min U \geq   u_N(x) \right) = \left( 1- 2^{-N} e^{-x + o(1)} \right)^{2^N} 
\end{equation}
for any $ x $ with the scaling function $ u_N(x) \coloneqq -\beta_c N + \frac{\ln(N \ln2) - \ln(4\pi)}{2\beta_c} - \frac{x}{\beta_c} $, where $\beta_c=\sqrt{2 \ln 2}$  and $ \mathbb{P} $ denotes the joint law of the REM. In particular, the minimal energy of the REM is roughly at $ - \beta_c N $. 

The trajectory observable we consider is (up to a factor of $t$) the empirical average of the REM energy along a trajectory  $ \pmb{\omega} $ of the Markov process
$$
U_t[ \pmb{\omega}] := \int_0^t U\left( \pmb{\omega}(s)\right) \ ds , \quad t > 0 .  
$$
We will be interested in the probability distribution of this quantity under the law $ \mathbb{P}_t $ on trajectories associated with $ W $ up to time $ t $ with the initial spin configurations equally distributed. The main result of this short note is a proof of a large deviation principle for this distribution in the limit of large system size $ N $ (for trajectories of any time extent $t$). This large deviation principle is  described in terms of the moment generating function 
\begin{equation}\label{eq:mom}
Z(t,\lambda) := 
	\int e^{-\lambda U_t[ \pmb{\omega}] } \ \mathbb{P}_t(d\pmb{\omega}) 
	= \sum_{\pmb{\sigma}, \pmb{\tau}\in \mathcal{Q}_N} 2^{-N} \langle  \pmb{\sigma} | e^{t \left( W - \lambda U \right)} | \pmb{\tau} \rangle 
	=\langle \fs |  e^{t \left( W - \lambda U \right)  } | \fs \rangle . 
\end{equation}
Here the second equality is due to the Feynman-Kac formula for the Markov process under consideration (cf.~\cite{KLW21,Leschke:2021xw}). Crucially, this formula connects the question concerning the (a)typical behavior of $ U_t $ to properties of the tilted generator 
$$ W_\lambda := W - \lambda U , $$
which is a random matrix on $ \ell^2(\mathcal{Q}_N) $. 
Note that by substituting $ W $ by $ \widehat W $, this random matrix coincides, up to a constant shift and change of sign, with the Hamiltonian of the QREM -- one of the simplest quantum spin glass models~\cite{Goldschmidt:1990kr,Manai:2020ta,Manai:2021nu,MaWa20}. In our case, the operator $H_\lambda = -W_\lambda + N \mathbbm{1}$ instead corresponds to the Hamiltonian with an all-to-all kinetic energy studied in~\cite{ASW15}. Due to the symmetry of the REM's distribution the parameter $ \lambda $ can be taken non-negative without loss of generality, and the large deviation function also known as scaled cumulant generating function (SCGF) is then given by
$$
\theta(t,\lambda) := \lim_{N\to \infty} \frac{1}{N t} \ln Z(t,\lambda) . 
$$
The SCGF plays the role of a free energy for trajectory ensembles.

It is important to emphasise that what we are considering here is 
very different from the study of classical thermal dynamics of the REM under Glauber or Metropolis schemes, as in e.g.\ \cite{Arous:2003lr,Arous:2003nt,Cerny:2017ie,Gayrard:2019yr}. In those cases the dynamical Markov generator is interacting (as transitions depend on changes in $U$) and what is studied are the typical trajectories under that interacting dynamics. In contrast we study rare trajectories under the non-interacting dynamics generated by $W$ with large fluctuations of $U_t$.

\section{Trajectory phase diagram}

Our main result is the following:
\begin{theorem}\label{thm:mainasym}
For any $ t  > 0, \lambda \geq  0 $ and almost all realisations of the REM:
\begin{equation}\label{eq:mainasym}
\theta(t,\lambda) = \max\left\{ 0 , t^{-1} p_0(t \lambda ) -1 \right\}  ,
\end{equation}
with 
\begin{equation}\label{eq:defthetanull}
	p_0( \beta ) := \begin{cases} 
		\frac{\beta^2}{2} , & \beta \leq  \beta_c := \sqrt{2 \ln 2} \\
		\beta \beta_c - \ln 2 , & \beta > \beta_c .
		\end{cases}
\end{equation}
\end{theorem}

Before spelling out the short proof of Theorem~\ref{thm:mainasym} in Section~\ref{sec:proof} below, let us put this result in some context and discuss some consequences. The quantity defined in~\eqref{eq:defthetanull} is the pressure corresponding to the REM's static (normalised) partition function at inverse temperature $ \beta $:
\begin{equation}\label{eq:REMp}
p_0( \beta )  = \lim_{N\to \infty} \frac{1}{N} \ln \frac{1}{2^N} \sum_{\pmb{\sigma}} e^{-\beta U(\pmb{\sigma}) }  .
\end{equation}
The critical value $ \beta_c = \sqrt{2 \ln 2} $ corresponds to the inverse of the REM's freezing temperature into a spin glass phase with 1-step replica symmetry breaking, cf.~\cite{Bov06}. 

The phase diagram resulting from Theorem~\ref{thm:mainasym} is thus composed of three regimes depicted in~Fig.~1(a):
\begin{enumerate}
\item  An  {\em Active} dynamical phase in which the Markov generator $ W $ dominates over the tilting, and which is characterised by $ \theta(t,\lambda) = 0 $ and the specific activity being unity (see below). It is separated from the remaining regimes by a first-order transition line.
This regime persists for all $ | \lambda | < \beta_c (2t)^{-1} + \beta_c^{-1} $ in case $ t^{-1} < 2 \beta_c^{-2} $ and $ |\lambda | <  \sqrt{2 t^{-1}} $ in case $ t^{-1} \geq  2 \beta_c^{-2} $. 
\item A regime of vanishing activity 
which occurs for $ t^{-1}  <2 |\lambda| \beta_c^{-1}  - 2  \beta_c^{-2}  $ and which is dominated by the REM's extreme values where the system localises. This regime is related to the spin-glass phase of the REM. 
We call this the {\em Inactive-1} dynamical phase. 

\item The remaining parameter regime corresponds to a second inactive regime which we term {\em Inactive-2} dynamical phase. 
It occurs only if $ t^{-1} > 2/\beta_c^2 $ and is related to the classical paramagnetic phase of the REM.
 \end{enumerate}

 In particular, in the long-time limit, $ t \to \infty $, the value $ \lambda = \beta_c^{-1} $ separates the Active and Inactive-1 phases, the latter dominating at at large $ \lambda $. Not surprisingly, this transition in the largest eigenvalue of the tilted generator $ W_\lambda= W - \lambda U $
 reflects the known location of its quantum-phase transition.
 As we will recall in Section~\ref{sec:proof} below, the eigenvector corresponding to the largest eigenvalue changes near $ \lambda = \beta_c^{-1} $ from a delocalised state resembling $ | \fs \rangle $
(indicating that trajectories visit all states equally giving rise to large activity)  to a state localised at the REM's maximising spin configuration (corresponding to trajectories that are inactive as they do not move away from this configuration). 

The classification above of the trajectory phases in terms of their activity is obtained as follows.
The dynamical activity is the total number of configuration changes in a trajectory. It can be calculated through the same tilting method used above for the time-integrated REM energy. Specifically, if we define the doubly tilted partition sum $Z(t,\lambda, s) := \langle \fs |  e^{t W_{\lambda,s}} | \fs \rangle$ with $W_{\lambda,s} = N e^{-s} | \fs \rangle \langle  \fs |  - N (1 - 2^{-N}(1-e^{-s}) ) \mathbbm{1} - \lambda U$ (where the additional tilting by $e^{-s}$ of the off-diagonal part of $W$ allows to count jumps in trajectories), we get the activity from $-\partial_s \log Z(t,\lambda, s)|_{s=0}$. Using the results above it is easy to see that the average activity per unit space and time is unity in the active phase and zero in the two inactive phases.

\begin{figure*}[t]
    \centering
    \includegraphics[width=\textwidth]{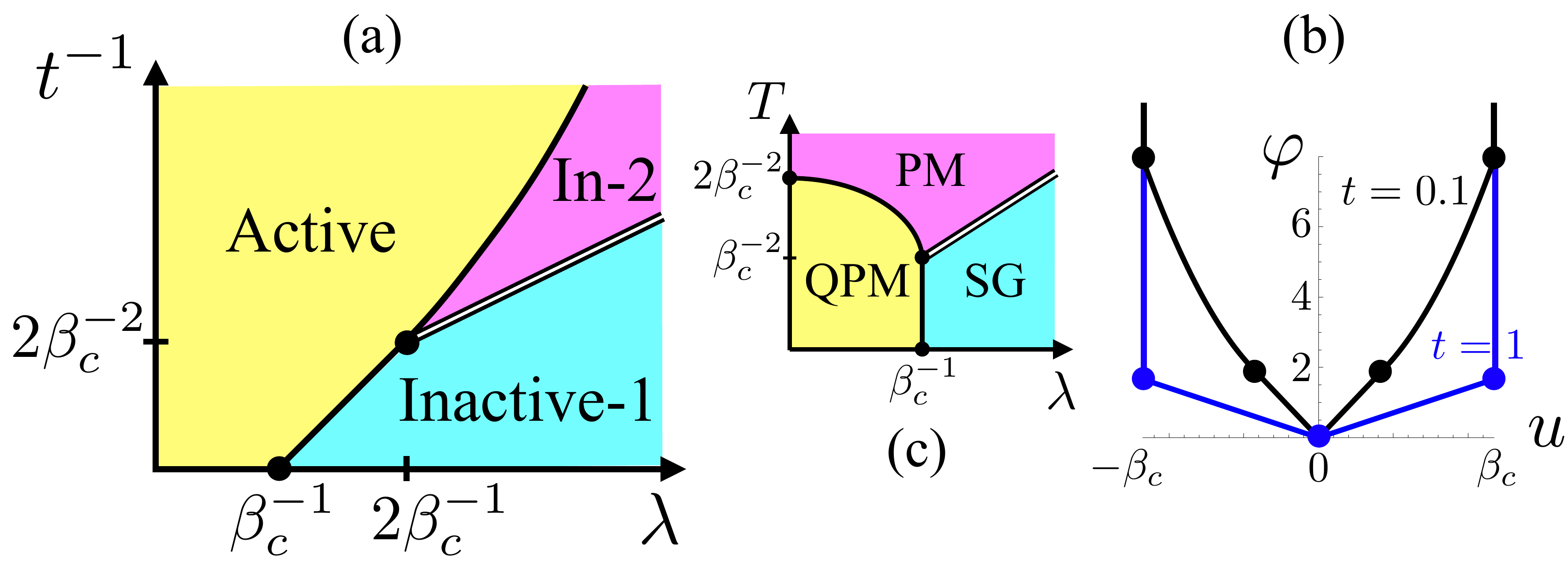}
	\caption{
	(a) Dynamical phase diagram in the limit of $N \to \infty$. The abscissa is the counting field $\lambda$
	conjugate to the time-integrated REM energy 
	and the ordinate the inverse of the trajectory length $t$. The full lines indicate first-order transitions between Active and Inactive-1 or Inactive-2 trajectory phases, while the double line indicates a 1-step RSB transition between the Inactive-1 and Inactive-2 phases. 
	(b) Large deviation function $\varphi(t,u)$ at two different values of $t$. For $t^{-1} = 1 < 2\beta_c^{-2}$ (blue lines) the rate function is one of coexistence between Active (which has $u=0$) and Inactive-1 (which has $u = \pm \beta_c$, by symmetry). The linear portion of the rate function is the Maxwell construction indicative of phase coexistence (in time). 
	For $t^{-1} = 10 > 2 \beta_c^{-2}$ the rate function describes the coexistence between the three phases (black). The linear portion between $0$ and $\sqrt{2t}$ is now the first-order coexistence between Active and Inactive-2. In Inactive-2 $|u|$ can take values with decreasing probability between $\sqrt{2t}$ and $\beta_c$. The rate function is infinite for any $|u|$ beyond $\beta_c$ (indicative of zero probability for such trajectories).
	(c) Thermal phase diagram of the all-to-all QREM for comparison to (a).}
	\label{fig:phase}
\end{figure*}

Via the G\"artner-Ellis theorem \cite{DemZeit98},  the rate function of the large deviation principle obeyed by $ U_t $ is given by the Legendre-Fenchel transformation 
$$
\varphi(t,u)  := \sup_{\lambda } \left( u \lambda -  \theta(t,\lambda) \right)  =
 \begin{cases} |u| \sqrt{\frac{2}{t}}, & |u| \leq \min\left\{ \sqrt{2t} , \beta_c \right\},  \\  1+ \frac{u^2}{2t}, & \mbox{else}, \\ \infty,  & |u| > \beta_c . \end{cases}
$$
Note that, although in  Theorem~\ref{thm:mainasym} initially  defined only for $ \lambda \geq 0 $ ,  the function $ \lambda \to  \theta(t,\lambda) $  extends to all real values by symmetry. The rate function $ u \to \varphi(t,u) $  is then symmetric as well. For times $ t > \beta_c^2/2 = \ln2 $, the second case in the above equation  is absent.
As a corollary to Theorem~\ref{thm:mainasym} and \cite[Thm 2.3.6]{DemZeit98}, we thus obtain the promised large deviation principle
\begin{align}\label{eq:LDPU}
 - \inf_{u \in I^\circ } t \varphi(t, u) & \leq    \liminf_{N\to \infty} \frac{1}{N} \ln \mathbb{P}_t\left((Nt)^{-1} U_t \in I \right) \notag \\ 
&  \leq  \limsup_{N\to \infty} \frac{1}{N} \ln \mathbb{P}_t\left( (N t)^{-1} U_t \in I  \right) = - \inf_{u \in \overline{I}} t \varphi(t, u) 
\end{align}
which holds for any Borel set $ I \subset \mathbb{R} $ and any $ t > 0 $. The rate function is shown in Fig.~1(b) for two different times. 

Clearly, under the apriori measure $ \mathbb{P}_t $, which favors rapid changes of spin configurations at the rate $ N(1-2^{-N}) $, the typical value of the REM's empirical energy density $N^{-1}  U_t[\pmb{\omega}] $ along any trajectory $ \pmb{\omega} $ is close to zero. 
The fluctuations about this typical behavior are described by~\eqref{eq:LDPU}: close to $ u = 0 $, these fluctuations are linearly suppressed with a rate proportionally to $ N \sqrt{ 2 t} $.  Tilting the apriori measure, one encounters one or two phase transitions depending on whether $ t >  \ln 2  $ or not. If $ t < \ln 2 $, one enters a regime  $ \sqrt{2t} <  | u | < \beta_c $ with Gaussian fluctuations. Beyond this, i.e., at energy densities of the order of the REM's maximum or minimum~\eqref{eq:maxU}, the energy density effectively stops fluctuating. Trajectories freeze for long times in the REM's extremal values.

\section{Comparison to the thermal phase diagram}

The dynamical partition sum of the stochastic system we are considering is reminiscent of a quantum (thermal and static) partition sum for the all-to-all version of the QREM. While the calculation of both is analogous, there are some important differences. Specifically, if we consider the tilted generator $W_\lambda$ as (minus) a Hamiltonian, the (specific) free energy of the associated quantum problem at temperature $T$ is 
\begin{equation}\label{eq:free}
f(T,\lambda) := \lim_{N\to \infty}  \frac{T}{N } \ln \frac{1}{2^N} \tr e^{( W - \lambda U)/T} = \max\left\{-  T \ln 2,  T p_0(\lambda/T)  -1 \right\} .  
\end{equation}
As we will explain in Section~\ref{sec:proof} below, the last equality follows straightforwardly from results on the eigenvalues in~\cite{ASW15}. 

Similarly to Theorem 1, from \eqref{eq:free} we see that, depending on coupling and temperature, the all-to-all QREM can be in three different phases, 
a delocalised quantum paramagnetic phase (QPM), a localised spin-glass phase (SG) and a classical paramagnetic phase (PM), see Fig.~1(c).
These three static quantum phases are similar to the dynamical ones of the stochastic problem. But is worth pointing out that at $T>0$ the (thermo)static phase transitions described by $ f(T,\lambda)  $, do not coincide with the dynamic phase transitions described by $ \theta(t,\lambda) $.
These differences arise because of the boundary vectors in the dynamical partition sum versus the trace in the static quantum one. For a comparison of the phase diagrams, see Figs.~1(a) and 1(c).

\section{Numerical illustration of finite size corrections}

The exact results above are for the limit $N \to \infty$. At finite $N$ there are of course finite-size corrections and sample-to-sample fluctuations between different realisations of the disorder $U$. Using numerics, we now illustrate some of these finite-size effects. (A comprehensive numerical study of both the all-to-all and single spin-flip problem will be presented in a future publication.)

\begin{figure*}[t]
    \centering
    \includegraphics[width=\textwidth]{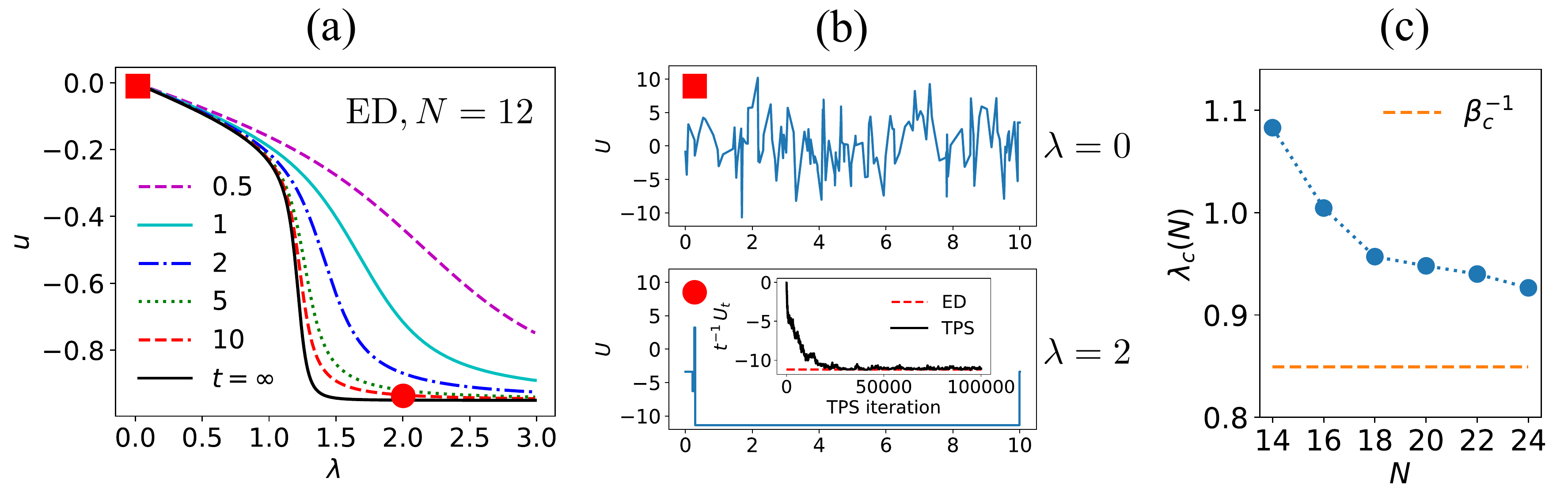}
	\caption{
	(a) Dynamical order parameter $-\partial_\lambda \theta(\lambda, t)$ for various times $t$, from exact diagonalisation (ED) for one disorder realisation with system size $N=12$. 
	(b) The top panel shows a typical trajectory of time extent $t=10$, corresponding to unbiased dynamics, $\lambda=0$ (generated from $W$ via standard continuous-time Monte Carlo from a random initial configuration). 
	The bottom panel shows a characteristic trajectory of the inactive phase, $\lambda=2$. This was obtained via transition path sampling (TPS, see main text). The inset to the lower panel shows the convergence of the trajectory sampling (black): each iteration is a different trajectory and their time-integrated energy converges to the ED result (red dashed) with enough TPS iterations. 
	(c) Transition point $\lambda_c(N)$ for $t=\infty$ as a function of system size, averaged over 20 disordered realisations. The dashed line is the large $N$ value of $\lambda_c$. 
	}
	\label{fig:traj}
\end{figure*}

When the system size is not too large the dynamical partition sum \eqref{eq:mom} can be computed numerically using exact diagonalisation (ED). We illustrate results for one disorder realisation in a system of size $N=12$ of the dynamical order parameter:
\begin{equation}\label{eq:u}
	u(\lambda, t) := \frac{1}{N t Z(t,\lambda)}  
		\int U_t[ \pmb{\omega}] e^{-\lambda U_t[ \pmb{\omega}] } \ \mathbb{P}_t(d\pmb{\omega}) 
		=
		-\frac{\partial}{\partial \lambda} \theta_N(t,\lambda)
\end{equation}
Figure 2(a) shows the following: (i) for finite size the phase transitions turn into crossovers, as expected; (ii) for all $t$ there is a crossover from $u \approx 0$ at $\lambda=0$ to a large negative $u$ for large $\lambda$, eventually reaching the minimum of the potential (which changes from sample to sample); (iii) these crossovers are sharper the longer $t$, also expected due to the preference of the boundary states in \eqref{eq:mom} for the delocalised state. 

In Fig.~2(b) we show representative trajectories for two values of $\lambda$ for $t=10$. We plot the instantaneous energy as a function of time in the trajectory. The top panel shows a typical trajectory of the dynamics corresponding to $\lambda=0$, cf.\ the red square in Fig.~2(a). This trajectory generated by $W$ is sampled using standard (continuous-time) Monte Carlo \cite{Bortz1975}. Since the unbiased dynamics connects all configurations with equal rates the trajectory jumps between the energy values: it corresponds to the phase which has high activity and is delocalised. The bottom panel shows a characteristic trajectory for $\lambda=2$, cf.\ the red circle in Fig.~2(a). This is a rare event (exponentially suppressed in $N$ and $t$) of the dynamics, and as such cannot be easily sampled from running  Monte Carlo with $W$ (since $W_\lambda$ is not a stochastic operator). We obtain such rare trajectories instead by performing importance sampling in trajectory space using transition path sampling (TPS) \cite{Bolhuis2002}, essentially a Monte Carlo method in trajectory space that aims to ``equilibrate'' to a reweigthed trajectory distribution 
$Z(t,\lambda)^{-1} e^{-\lambda U_t[ \pmb{\omega}] }$ (supplemented with bridge moves to improve acceptance; we will provide details of this method in a future publication). The inset to the lower panel shows the convergence of our TPS approach: it shows the evolution of the sampled trajectories with TPS iterations by showing their $U_t[ \pmb{\omega}]$ (per unit time). The $U_t[ \pmb{\omega}]$ in the inset converges eventually to the value expected at $\lambda=2$, showing that TPS converges to the tilted trajectory ensemble. The trajectory shown in the lower panel is the last trajectory from TPS. It is very different from the typical one in the upper panel: it has very low activity and is localised for most of the time in the minium energy configuration, corresponding to the Inactive-1 dynamical phase. Note that while we only illustrate the numerics for the size $N=12$, TPS can be used for larger system sizes in contrast to ED. 

Figure 2(c) shows the location of the critical $\lambda$  for $t=\infty$, averaged over 20 realisations of the disorder, for different systems sizes. The transition point is inferred from the maximum of the {\em dynamical susceptibility} $\partial_\lambda^2 \theta_N(\lambda)$, where $\theta_N(\lambda)$ is the largest eigenvalue of $W_\lambda$. This eigenvalue is calculated using \eqref{eq:chareq} below, which allows to compute it for larger sizes than those accessible to ED. The figure suggest a convergence to the limiting value $\beta_c^{-1}$ for large $N$, as expected from the analytics above.

\section{Proof of the large deviation result}\label{sec:proof}
The Feynman-Kac formula~\eqref{eq:mom} reduces the large deviation problem to a spectral analysis of the random matrix 
$ W_\lambda = W - \lambda U $, which --  as motivated in the introduction -- may serve as a toy model to the QREM with a simpler all-to-all kinetic energy term with. 
Up to a constant shift and rescaling, the spectrum of $ W_\lambda $ has been analysed in~\cite{ASW15} both on the macro and microscopic scale of the eigenvalue process. 
The main technical tool for studying $ W_\lambda$ is rank-one perturbation theory according to which $ E $ is an eigenvalue of $ W_\lambda $ if and only if 
\begin{equation}\label{eq:chareq}
\frac{1}{N} = \langle \fs | (E+N + \lambda U)^{-1} | \fs \rangle = \frac{1}{2^N} \sum_{\pmb{\sigma}} \frac{1}{E+N +\lambda U(\pmb{\sigma})} .
\end{equation}
The corresponding eigenvectors $\psi_E $  satisfy for all $ \pmb{\sigma},  \pmb{\tau} \in \mathcal{Q}_N $:
\begin{equation}\label{eq:ratio}
	\frac{\langle \pmb{\sigma} | \psi_E \rangle}{\langle \pmb{\tau} | \psi_E \rangle} = \frac{E + N + \lambda U(\pmb{\tau}) }{E+N +\lambda U(\pmb{\sigma})} . 
\end{equation}
An immediate implication of~\eqref{eq:chareq} is the fact that all eigenvalues of $ W_\lambda $ aside from the largest one are interlaced with the REM's energies and additionally shifted by $- N $ (cf.\ e.g.\ \cite{AizWar15} and refs.\ therein for interlacing and finite-rank perturbation theory). All eigenvalues are almost surely simple. Moreover, any
 solution of~\eqref{eq:chareq} with $ E > ( \lambda \beta_c -1) N $ (cf.~\eqref{eq:maxU}) is independent of the realisation of $ U $ up to exponentially small fluctuations. By the law of large numbers the 
right-hand side of~\eqref{eq:chareq} is then well approximated as $ N \to \infty $ by the integral
$$
\int_{-\infty}^\infty \frac{\exp(-v^2/2) }{E + N + \lambda \sqrt{N} v} \frac{dv}{\sqrt{2\pi}} = \frac{1}{E + N} \left( 1 + \frac{\lambda^2}{N (1+ E/N) } +  \mathcal{O}(N^{-2}) \right) . 
$$
This explains the following results on the largest eigenvalue $ E_0 := \max \sigma(W_\lambda) $, which are found in~\cite{ASW15}:
\begin{enumerate}
\item In case $ \lambda \beta_c < 1 $, on an event with probability exponentially close to one, the largest eigenvalue is at $E_0 := \max \sigma(W_\lambda) =  \lambda^2 + \mathcal{O}(N^{-1}) $ and the corresponding eigenvector satisfies
$ \langle \pmb{\sigma} | \psi_{E_0} \rangle \propto (E_0 + N  + \lambda U(\pmb{\sigma}) )^{-1}  $. Since $ 2^{-N} \sum_{\pmb{\sigma} } (E_0 + N  + \lambda U(\pmb{\sigma}) )^{-2}  $ is of order one up to exponentially small fluctuations by the law of large numbers, this vector is hence still delocalised (as in the case $ \lambda = 0 $). 
\item  In case $ \lambda \beta_c > 1 $, the largest eigenvalue is at 
\begin{equation}\label{eq:locE0}
E_0 = \max \sigma(W_\lambda) = - \lambda \min U - N  - N 2^{-N} \left( 1 - \frac{1}{\lambda \beta_c} \right)^{-1} + o(N 2^{-N} ) 
\end{equation}
and the corresponding eigenvector is mostly concentrated on the REM's minimising configuration $ \pmb{\sigma}_0 \coloneqq \arg \min U $. This is specified through the ratios~\eqref{eq:ratio}. Note that the error term in the above equation only holds with a probability up to $ 1 - \mathcal{O}(1/N) $, cf.~\cite[App A]{ASW15}. 
\end{enumerate}
In particular, the union of eigenvalues, $ \sigma(W_\lambda) $, when divided by $ N $, converges almost surely to the non-random set $ \{ 0\} \cup [- \lambda \beta_c -1, \lambda \beta_c-1]  $. Together with the interlacing property, one then also easily arrives at~\eqref{eq:free} for the free energy of $ W_\lambda $.\\

The proof of Theorem~\ref{thm:mainasym} requires slightly more detailed knowledge, since $ \langle \fs | e^{t W_\lambda} | \fs \rangle $ involves properties of the eigenvectors, too. The rough picture established in~\cite{ASW15} through a more detailed
analysis of the characteristic equation~\eqref{eq:chareq} is the following: 
\begin{enumerate}
\item Delocalisation of one eigenstate near energy $ 0 $ is shown to persist up to  $ \lambda < \sqrt{2} $. From that value on, $ \lambda > \sqrt{2} $, this eigenstates ``melts'' into a narrow band of semi-delocalised states near energy $ 0 $. 
\item The eigenvalue process, when rescaled to order one at some fixed energy outside $ - N $ and $ 0 $,  is given by a Poisson process. Correspondingly, outside those special energies the normalised eigenvectors are localised. 
\end{enumerate}
We will need the following result, which is contained in~\cite[Proof of Thm. 6.3]{ASW15}. 
\begin{proposition}\label{prop:ASW15}
For any $ \delta > 0 $ and any $ N $ there is some $ a > 0 $ and an event $ \Omega_N $ whose complement is summable, $ \sum_N \mathbb{P}\left( \Omega_N^c\right) < \infty $, such that in the event $ \Omega_N $ any eigenvalue $ E $ of $ W_\lambda $ with $ |E| > \delta N $ and $ |E + N | > \delta N $ has a normalised eigenvector $ \psi_E $, which satisfies 
$| \langle \fs | \psi_E \rangle |^2 \leq N^a \ 2^{-N } $.  Moreover, for any such $ E $, there is some $ \pmb{\sigma}_E \in \{ - 1, 1 \}^N $ such that $ | E +N - \lambda U( \pmb{\sigma}_E) | \leq \delta N $.
\end{proposition}

\begin{proof}[Proof of Theorem~\ref{thm:mainasym}]
The proof proceeds by establishing asymptotically coinciding upper and lower bounds. 
For the lower bound, we use Jensen's inequality to conclude
$$
 \ln \langle \fs | e^{t W_\lambda} | \fs \rangle \geq t  \langle \fs | W_\lambda | \fs \rangle  =  \frac{t \lambda}{2^N} \sum_{\pmb{\sigma}} U(\pmb{\sigma}) .
$$
By the law of large numbers, this term converges to zero for almost all realisations of the REM.  For another lower bound, which is sharper in case $  t < p_0(t\lambda)  $, we estimate
\begin{align*}
\langle \fs | e^{t W_\lambda} | \fs \rangle = \frac{1}{2^N} \sum_{\pmb{\sigma}, \pmb{\tau}}  \langle \pmb{\tau}  | e^{t W_\lambda} | \pmb{\sigma} \rangle & \geq  \frac{1}{2^N} \sum_{\pmb{\sigma}}  \langle \pmb{\sigma}  | e^{t W_\lambda} | \pmb{\sigma} \rangle \\
& \geq  \frac{1}{2^N} \sum_{\pmb{\sigma}}  \exp\left( tN ( |\langle \pmb{\sigma} | \fs \rangle |^2 -1)- t \lambda U(\pmb{\sigma}) \right) ,
\end{align*}
where the last step is again by Jensen's inequality.  Using $ |\langle \pmb{\sigma} | \fs \rangle|^2 = 2^{-N} $ and~\eqref{eq:REMp}, the combination of the above estimates yields~\eqref{eq:mainasym} as a lower bound. 

A complementing upper bound is based on Proposition~\ref{prop:ASW15}. Expanding in eigenfunctions and splitting the sum over all eigenvalues in three parts corresponding to energies $ E $ with $ |E + N| \leq \delta N $, $ | E | \leq \delta N $ and the rest, we  write and estimate using Proposition~\ref{prop:ASW15}:
\begin{align}\label{eq:upperb}
 \langle \fs | e^{t W_\lambda} | \fs \rangle = \sum_E e^{t E} | \langle \fs | \psi_E \rangle |^2 \leq e^{tN(\delta-1)} + e^{t N\delta} + \frac{N^a}{2^N} e^{tN(\delta-1)} \sum_{\pmb{\sigma}} e^{-t \lambda U(\pmb{\sigma})} .
\end{align}
In the event $ \Omega_N $ of  Proposition~\ref{prop:ASW15}, we thus conclude
$$
\limsup_{N\to \infty} \frac{1}{N t } \ln \langle \fs | e^{t W_\lambda} | \fs \rangle \leq \max\left\{ 0 , t^{-1} p_0(t\lambda) - 1 \right\} + \delta .
$$
By a Borel-Cantelli argument, this establishes this almost-sure bound on the upper limit. Since $ \delta > 0 $ is arbitrary, this concludes the proof.
\end{proof}

\section{Outlook: QREM}
Let us conclude this note with some conjectures, partial results and comparison in case $ W $ is replaced by the spin-flip dynamics generated by $ \widehat W $. 
In that case, the tilted generator $ \widehat H_\lambda :=  \widehat W - \lambda U $ is the QREM. Its low-energy spectrum as well as the phase transitions in the free energy are well understood \cite{Manai:2020ta,MaWa22}. By the Feynman-Kac formula the dynamical phase transition is again described in terms of the asymptotic behavior of $N^{-1}  \ln \langle \fs | e^{t \widehat H_\lambda} | \fs \rangle $.

The phase transition in the largest eigenvalue $ \widehat E_0 := \max \sigma(  \widehat H_\lambda) $ occurs on order $ N $ at the same location $ \lambda = \beta_c^{-1} $ as for $ H_\lambda $. However, the finite-volume corrections are different in the localisation regime, i.e.\ for all realisations of the REM aside from a set of exponentially small probability (see~\cite{MaWa22} for details):
\begin{enumerate}
\item if $ \lambda > \beta_c^{-1} $ we have $ \widehat E_0 =- \lambda \min U + (\lambda \beta_c)^{-1} + \mathcal{O}(N^{-1/4}) $, 
\item if $ \lambda < \beta_c^{-1} $ we have $ \widehat E_0=  \lambda^2 + \mathcal{O}(N^{-1/4}) $.
\end{enumerate}
Following the steps of the lower bound in the proof of Theorem~\ref{thm:mainasym}, it is easy to see that for almost all realizations of the REM one still has:
\begin{equation}\label{eq:lowerQREM}
\liminf_{N\to \infty} \frac{1}{N t } \ln \langle \fs | e^{t \widehat H_\lambda} | \fs \rangle \geq \max\left\{ 0 , t^{-1} p_0(t\lambda) - 1 \right\} .
\end{equation}
We conjecture that this bound is sharp. In fact, using the spectral decomposing as in~\eqref{eq:upperb} and decomposing the sum into positive and negative energies we may again estimate
$$
 \langle \fs | e^{t \widehat H_\lambda} | \fs \rangle \leq \sum_{\substack{ E \sigma(  \widehat H_\lambda)\\ E > 0 } } e^{t E} \left| \langle \fs | \psi_E \rangle \right|^2 + 1 .
$$
The first sum is estimated trivially by $ e^{t \widehat E_0 } $. In case $  \lambda < \beta_c^{-1} $ this yields the upper bound $ \limsup_{N\to \infty} \frac{1}{N t } \ln \langle \fs | e^{t \widehat H_\lambda} | \fs \rangle \leq 0 $, which coincides with the lower bound. In case  $  \lambda > \beta_c^{-1} $, we know from~\cite{MaWa22} that eigenvalues with energies $ E > 0 $ are in one-to-one correspondence with values $ U(\pmb{\sigma}_E) = E + \mathcal{O}(1) $. We conjecture that the local density of states at these energies satisfies $ \lim_{N\to \infty} N^{-1} \ln \langle \fs | 1_{(E-\delta_N,E+\delta N)}(  \widehat H_\lambda ) | \fs \rangle = - \ln 2 $ for all sufficiently small $ \delta > 0 $. This would prove that~\eqref{eq:lowerQREM} is indeed sharp.   
 \section*{Acknowledgements}
 JPG acknowledges financial support from EPSRC Grant no. EP/R04421X/1 and the Leverhulme Trust Grant No. RPG-2018-181. SW thanks the DFG for support under grant EXC-2111 -- 390814868. Numerical simulations were performed using the Sulis Tier 2 HPC platform funded by EPSRC Grant EP/T022108/1 and the HPC Midlands+ consortium. \\
 
 \noindent
{\bf Data:} \; Research data are available from the Nottingham Research Data Management Repository at http://doi.org/10.17639/nott.7196.

\end{document}